\documentclass[psamsfonts,reprint]{amsart}
\usepackage{amssymb,amsfonts}
\usepackage[all,arc]{xy}
\usepackage{enumerate}
\usepackage{mathrsfs}
\usepackage{graphicx}
\usepackage{enumitem}
\usepackage{bm}
\usepackage{array,arydshln}
\usepackage{extarrows}
\usepackage[table]{xcolor}
\usepackage{mathptmx}
\usepackage{yhmath}
\usepackage{abstract}

\newtheorem{thm}{Theorem}[section]
\newtheorem{cor}[thm]{Corollary}

\newtheorem{lem}[thm]{Lemma}

\newtheorem{quest}[thm]{Question}
\newtheorem{coj}{Conjecture}[section]

\theoremstyle{definition}
\newtheorem{defn}[thm]{Definition}

\newtheorem{exmp}[thm]{Example}

\theoremstyle{remark}
\newtheorem{rem}[thm]{Remark}

\makeatletter
\let\c@equation\c@thm
\makeatother
\numberwithin{equation}{section}
\title{All 2-positive linear maps from $M_3(\mathbb{C})$ to $M_3(\mathbb{C})$ are decomposable}

\author{Yu Yang}
\address[Yu Yang]{Department of
Mathematics, National University of Singapore, 10 Lower Kent Ridge Road, Singapore 119076, Republic of Singapore}
\email{a0086285@nus.edu.sg}
\author{Denny H. Leung}
\address[Denny H. Leung]{Department of
Mathematics, National University of Singapore, 10 Lower Kent Ridge Road, Singapore 119076, Republic of Singapore}
\email{matlhh@nus.edu.sg}
\author{Wai-Shing Tang}
\address[Wai-Shing Tang]{Department of
Mathematics, National University of Singapore, 10 Lower Kent Ridge Road, Singapore 119076, Republic of Singapore}
\email{mattws@nus.edu.sg}

\date{\today}

\begin{document}

\maketitle

\begin{abstract}
Following an idea of Choi, we obtain a decomposition theorem for $k$-positive linear maps from $M_m(\mathbb{C})$ to $M_n(\mathbb{C})$, where $2 \leq k < \min\{m, n\}$.
As a consequence, we give an affirmative answer to Kye's conjecture (also solved independently by Choi) that every 2-positive linear map from $M_3(\mathbb{C})$ to $M_3(\mathbb{C})$ is decomposable.
\end{abstract}


\section{Introduction}
Let $M_{n}(\mathbb{C})$ be the algebra of all $n\times n$ matrices over the complex field $\mathbb{C}$.
We say that a matrix $A$ in $M_{n}(\mathbb{C})$ is positive semi-definite, and write $A \ge 0$, if $A$ is hermitian and all eigenvalues of $A$ are non-negative.
Denote by $M_{n}^{+}(\mathbb{C})$ the set of all positive semi-definite matrices in $M_{n}(\mathbb{C})$,
and by $B(M_{m}(\mathbb{C}), M_{n}(\mathbb{C}))$ the space of all linear maps from $M_{m}(\mathbb{C})$ to $M_{n}(\mathbb{C})$.

\begin{defn}
  A linear map $\phi$ from $M_{m}(\mathbb{C})$ to $M_{n}(\mathbb{C})$ is called positive if $\phi(M_{m}^{+}(\mathbb{C}))\subseteq M_{n}^{+}(\mathbb{C})$.
\end{defn}
%The identity map on $M_{n}(\mathbb{C})$ is labelled as $id_n$ and the transpose map on $M_{n}(\mathbb{C})$ is labelled as $\tau_n$, respectively.
The identity map on $M_{n}(\mathbb{C})$ and the transpose map on $M_{n}(\mathbb{C})$ are denoted by $id_n$ and $\tau_n$ respectively.
\begin{defn}
  A map $\phi$ is called $k$-positive if the map $id_{k}\otimes \phi: M_{k}(M_m(\mathbb{C}))\rightarrow M_{k}(M_n(\mathbb{C}))$ is positive. Similarly, a map $\phi$ is called $k$-copositive if the map $\tau_{k}\otimes \phi:M_{k}(M_m(\mathbb{C}))\rightarrow M_{k}(M_n(\mathbb{C}))$ is positive.
\end{defn}
If a map is $k$-positive (resp. $k$-copositive) for every $k$, it is called completely positive (resp. completely copositive). A positive map is called decomposable
if it can be written as the sum of a completely positive map and a completely copositive map.

In \cite{gchoi}, Cho, Kye and Lee introduced the generalized Choi maps and discussed the conditions for the maps to be $k$-positive or decomposable. For generalized Choi's map in $B(M_{3}(\mathbb{C}),M_3(\mathbb{C}))$, they showed that 2-positivity or 2-copositivity implies decomposability. It is natural to ask whether this property holds for every 2-positive or 2-copositive map in $B(M_{3}(\mathbb{C}),M_3(\mathbb{C}))$
(see \cite[page 1330002-11]{facial}).

\begin{coj}
Every 2-positive (respectively 2-copositive) linear map in $B(M_{3}(\mathbb{C}),M_3(\mathbb{C}))$ is decomposable.
\end{coj}

%Later in \cite{3by3},

Let us recall some useful definitions.
Denote by $\mathcal{H}_A$ and $\mathcal{H}_B$ two Hilbert spaces with $dim(\mathcal{H}_A)=m$ and $dim(\mathcal{H}_B)=n$, respectively.
\begin{defn}
  Every vector $z \in \mathcal{H}_A \otimes \mathcal{H}_B$ has a canonical expansion $z = \sum_{i=1}^{m} e_i \otimes z_i$,
  where $\{e_i\}_{i=1}^{m}$ is a basis for $\mathcal{H}_A$ and $z_i \in \mathcal{H}_B$ for $i=1,2,...,m$.
  The {\em Schmidt rank} $SR(z)$ of the vector $z$ is defined to be the dimension of $span(\{z_1,...,z_m\})$.
\end{defn}

\begin{defn} \cite{TH, 3by3}
%For a quantum state,
Consider the density matrix $\rho$ for a quantum state in a bipartite system $\mathcal{H}_A\otimes\mathcal{H}_B$.
The {\em Schmidt number} of the density matrix (or the state) $\rho$ is defined by
$$ SN(\rho)=\min\limits_{} \bigg\{\max\limits_{k} SR(z_k)\bigg\}, $$
where the minimum is taken over all possible decompositions
$$\rho=\sum_{k}^{}p_k\cdot z_k z_k^*$$
with $z_k$ being vectors in $\mathcal{H}_A\otimes\mathcal{H}_B$ and $p_{k}> 0, \ \sum_{k} {p_{k}} = 1$.
\end{defn}

%Moreover,
Sanpera, Bru{\ss} and Lewenstein in \cite{3by3} formulated the following conjecture and presented strong evidence of its validity for some special cases.
%(see \cite{3by3}).

\begin{coj}
  All bound entangled states with positive partial transpose in $\mathbb{C}^3\otimes\mathbb{C}^3$ have Schmidt number 2.
\end{coj}

There is a diagram of dual cone relations between quantum states and positive maps (see \cite{scone,itoh,mk,facial}). Let us consider the duality between the space $M_m(\mathbb{C})\otimes M_n(\mathbb{C})$ and the space $B(M_{m}(\mathbb{C}),M_n(\mathbb{C}))$. Let $E_{ij}$ be the canonical matrix units in $M_m(\mathbb{C})$. For $A=\sum_{i,j=1}^{m}E_{ij}\otimes A_{ij}\in M_m(\mathbb{C})\otimes M_n(\mathbb{C})$ and a linear map $\phi\in B(M_{m}(\mathbb{C}),M_n(\mathbb{C}))$, define a bilinear form:
$$ \langle A,\phi\rangle \, = \, \sum_{i,j=1}^{m}Tr(\phi(E_{ij})A_{ij}^t).$$

Note that for two normed real spaces $X$ and $Y$ which are dual to each other with respect to a bilinear form $\langle\cdot,\cdot\rangle$,
the dual cone for a subset $C$ of $X$ is defined as $C^{\circ}=\{y\in Y:\quad \langle x,y\rangle\geq 0 \text{ for each }x\in C\}$. Denote by $\mathbb{P}_k[m,n]$ and $\mathbb{P}^k[m,n]$ the set of all $k$-positive maps and the set of all $k$-copositive maps in $B(M_{m}(\mathbb{C}),M_n(\mathbb{C}))$, respectively. Define convex cones $\mathbb{V}_k[m,n]$ and $\mathbb{V}^k[m,n]$ in $M_m(\mathbb{C})\otimes M_n(\mathbb{C})$ as
\begin{align*}
  \mathbb{V}_k[m,n]&=\{zz^*: \; SR(z)\leq k, \; z \text{ in } \mathbb{C}^m\otimes\mathbb{C}^n\}^{\circ\circ} ,\\
  \mathbb{V}^k[m,n]&=\{(zz^*)^{\Gamma}: \; SR(z) \leq k, \; z \text{ in } \mathbb{C}^m\otimes\mathbb{C}^n\}^{\circ\circ}.
\end{align*}

Here $\Gamma$ is an operation called partial transposition that acts as transposition only on the first part of a tensor product.
By the dual correspondence between maps and states, we have the following diagram:
$$\begin{array}{cccccccl}
\mathbb{V}_1 & \subsetneqq  & \cdots &  \mathbb{V}_k & \subsetneqq &\mathbb{V}_{m\wedge n} & = &(M_m(\mathbb{C})\otimes M_n(\mathbb{C}))^+\\
\updownarrow &              &         &   \updownarrow &            & \updownarrow                                                        \\
\mathbb{P}_1 & \supsetneqq  & \cdots &  \mathbb{P}_k   &\supsetneqq & \mathbb{P}_{m\wedge n}& \cong & (M_m(\mathbb{C})\otimes M_n(\mathbb{C}))^+\\
\end{array},$$
where $m\wedge n \, = \, \min\{m,n\}$. A similar diagram holds in case of copositivity:
$$\begin{array}{cccccccl}
%\mathbb{V}^1 & \subsetneqq  & \cdots &  \mathbb{V}^k & \subsetneqq &\mathbb{V}^{m\wedge n} &     & (M_m(\mathbb{C})\otimes M_n(\mathbb{C}))^+\\
\mathbb{V}^1 & \subsetneqq  & \cdots &  \mathbb{V}^k & \subsetneqq &\mathbb{V}^{m\wedge n} &     &                                   \\
\updownarrow &              &          & \updownarrow &            & \updownarrow                                                     \\
\mathbb{P}^1 & \supsetneqq  & \cdots &  \mathbb{P}^k &\supsetneqq & \mathbb{P}^{m\wedge n} &\cong& (M_m(\mathbb{C})\otimes M_n(\mathbb{C}))^+\\
\end{array}.$$
Denote by $\mathbb{D}[m,n]$ the convex cone given by $\mathbb{P}_{m\wedge n}+\mathbb{P}^{m\wedge n}$.
Correspondingly, denote by $\mathbb{T}[m,n]$ the cone of states given by $\mathbb{V}_{m\wedge n}\bigcap \mathbb{V}^{m\wedge n}$.

One can check that $(\mathbb{D}[m,n],\mathbb{T}[m,n])$ is a dual pair defined through the bilinear pairing between $B(M_{m}(\mathbb{C}),M_n(\mathbb{C}))$ and $M_m(\mathbb{C})\otimes M_n(\mathbb{C})$. It is natural to ask where should we locate the pair $(\mathbb{D}[m,n],\mathbb{T}[m,n])$ in the above two diagrams.
Moreover, it follows from duality in the diagrams that Conjecture 1.1 and Conjecture 1.2 are equivalent.

%In this paper, we give an affirmative answer to Conjecture 1.1 and hence Conjecture 1.2 too (see Theorem 3.2 and Corollary 3.6).

This paper is organized as follows.  In Section~2, we will give a decomposition theorem in order to relate a $k$-positive map in $B(M_m(\mathbb{C}),M_n(\mathbb{C}))$ to a $(k-1)$-positive map which actually resides in $B(M_{m-1}(\mathbb{C}),M_n(\mathbb{C}))$.
In Section~3, we will give an affirmative answer to Conjecture 1.1 and hence Conjecture 1.2 too (see Theorem 3.2 and Corollary 3.6).
%Several examples are illustrated throughout these two sections.
Examples are provided throughout to illustrate certain aspects of the results obtained in these sections.

\section{A decomposition for all $k$-positive / $k$-copositive maps}
 Our approach towards Conjecture~1.1 is to peel off a completely positive map from a 2-positive map. That is, find a completely positive map which is dominated by the 2-positive map. Moreover, the dimension of the space where the remaining map resides is reduced. Indeed, this is a dimension-lowering trick.
 %From now on, the field is confined to be $\mathbb{C}$ and all spaces involved are finite dimensional.
 For a positive linear map on matrix algebras, a classical theorem by Choi \cite{choi} is important in determining complete positivity.
 Before that, we recall the notion of the Choi matrix for a linear map.
\begin{defn}
  Let $B(K)$ and $B(H)$ denote the space of bounded linear operators on finite dimensional Hilbert spaces $K$ and $H$, respectively.
  Let $E_{ij},i,j=1,...,m$, be the canonical matrix units for $B(K)$ and $(dim(K),dim(H))=(m,n)$. Given a linear map $\phi\in B(B(K),B(H))$, the Choi matrix $C_{\phi}$ for $\phi$ is:
$$C_{\phi}\triangleq \sum_{i,j=1}^{m}E_{ij}\otimes \phi(E_{ij})=[\phi(E_{ij})]_{i,j=1}^{m}\in M_{m}(M_{n}(\mathbb{C})).$$
\end{defn}
\begin{rem}
  Obviously the map $\phi\longmapsto C_{\phi}$ is a bijection between linear maps in $B(M_{m}(\mathbb{C}),M_{n}(\mathbb{C}))$ and matrices in $M_{m}(M_n(\mathbb{C}))$ which preserves linearity.
\end{rem}

\begin{thm}[Choi,1975]
  A positive map $\phi\in B(B(K),B(H))$ is completely positive if and only if
  the corresponding Choi matrix is positive.
\end{thm}

The peel-off theorem first appeared in \cite{peel} (see also St{\o}rmer's book \cite[pages 38-39]{sbook}). Combined with Zorn's Lemma St{\o}rmer obtained a decomposition for positive maps in \cite{s2}.
Here we present a slightly stronger version (Theorem~2.7) of the peel-off result by block-matrix approach, which was shown by M.-D. Choi for the case of 2-positive maps \cite{choi2}.
Let us consider $k$-positive maps for the moment. A similar theorem holds for $k$-copositive maps.

\begin{defn}[Trivial Lifting]
  Given a linear map $\chi\in B(M_{s}(\mathbb{C}),M_n(\mathbb{C}))$, fix the canonical matrix unit basis $E_{ij}$, $i,j=1,..,s$, in $M_{s}(\mathbb{C})$, under which the Choi matrix is $C_{\chi}=[\chi(E_{ij})]_{i,j=1}^{s}\in M_{s}(M_n(\mathbb{C}))$. Given $I=\{n_1,...,n_p\}\subset\{1,...,s+p\}$, where $n_1<\cdots<n_p$,  extend the matrix $C_{\chi}$ to a $(s+p)\times (s+p)$ block matrix $C^{lift}_{I}\in M_{s+p}(M_n(\mathbb{C}))$ by adding one row and one column of $n\times n$ zero matrices at the $n_k^{th}$ level for each $k=1,...,p$ as follows:

\begin{align*}
	C^{lift}_{I}&\triangleq
	\bordermatrix{     & 1^{st}  & \cdots & n_k^{th} & \cdots & (s+p)^{th} \cr
                     1^{st} & \chi(E_{11}) & \cdots & 0 & \cdots & \chi(E_{1,s})\cr
                \vdots & \vdots & \ddots & 0 & \ddots & \vdots \cr
                     n_k^{th} & 0 & 0 & 0 & 0 & 0\cr
                \vdots & \vdots & \ddots & 0 & \ddots &  \vdots\cr                    (s+p)^{th} & \chi(E_{s,1}) & \cdots & 0 & \cdots & \chi(E_{s,s})  \cr
            }.
\end{align*}
Denote by $\tilde{\chi}_{I}$ the map in $B(M_{s+p}(\mathbb{C}),M_n(\mathbb{C}))$ associated with the Choi matrix $C_{\tilde{\chi}_I}=[\tilde{\chi}_{p}(E_{ij})]_{i,j=1}^{s+p}=C^{lift}_{I}$. Then the map $\tilde{\chi}_{I}$ is called a I-trivial lifting of the original map $\chi$. If $I=\{q\}$ is a singleton, simply denote by $\tilde{\chi}_q$ the $q$-trivial lifting of $\chi$.
\end{defn}

\begin{lem}
  The map $\chi$ is k-positive or k-copositive if and only if the trivial lifting $\tilde{\chi}_p$ is k-positive or k-copositive, respectively.
\end{lem}
\begin{proof}
  Let $\eta=(w^1,...,w^{k})^{T}$ be an arbitrary vector in $\mathbb{C}^k\otimes \mathbb{C}^m$ where $w^s\in\mathbb{C}^m,\ s=1,...,k$. Let $\hat{w}^s\in \mathbb{C}^{m-1}$ be defined as $(w_{1}^s,...,w_{p-1}^s,w_{p+1}^s,...,w_m^s)^{T}$ for $s=1,...,k$, and $\hat{\eta}=(\hat{w}^1,...,\hat{w}^k)\in\mathbb{C}^k\otimes\mathbb{C}^{m-1}$. By definition of p-trivial lifting, $$(id_k\otimes \tilde{\chi}_p)(\eta\eta^*)=[\tilde{\chi}_p(w^s(w^t)^*)]_{s,t=1}^k=[\chi(\hat{w}^s(\hat{w}^t)^*)]_{s,t=1}^k=(id_k\otimes\chi)(\hat{\eta}\hat{\eta}^*).$$ This matrix equality in $M_{k}(M_{n}(\mathbb{C}))$ shows that the pair of maps $(\chi,\tilde{\chi}_p)$ are $k$-positive simultaneously. For $k$-copositivity, we also have: $$(\tau_k\otimes \tilde{\chi}_p)(\eta\eta^*)=[\tilde{\chi}_p(w^t(w^s)^*)]_{s,t=1}^k=[\chi(\hat{w}^t(\hat{w}^s)^*)]_{s,t=1}^k=(\tau_k\otimes \tilde{\chi})(\hat{\eta}\hat{\eta}^*),$$
  which completes the proof.
\end{proof}

\begin{rem}
  By repeatedly using Lemma 2.5, a map $\chi$ is k-positive or $k$-copositive, if and only if the trivial lifting $\tilde{\chi}_{I}$ is $k$-positive or $k$-copositive, respectively.
\end{rem}

\begin{thm}(Choi Decomposition)
  Let $\phi$ be a non-zero k-positive $(2 \leq k < \min\{m,n\})$ map in $B(M_m(\mathbb{C}),M_n(\mathbb{C}))$. Then there exists a decomposition $\phi=\psi+\gamma$, where $\psi$ is a non-zero completely positive map and $\gamma$ is a p-trivial lifting of a $(k-1)$-positive map in $B(M_{m-1}(\mathbb{C}),M_n(\mathbb{C}))$, for some $p\in\{1,...,m\}$.
\end{thm}
Before proving Theorem 2.7, recall a classical result (see \cite[Exercise 1.3.5]{prin}):
\begin{lem}
   Suppose a hermitian matrix M is partitioned as
    \begin{align*}
      M=\begin{pmatrix}A&B\\ B^*&C\end{pmatrix},
    \end{align*}
    where A and C are square matrices. Then the following conditions are equivalent:
\begin{enumerate}
	\item $M\geq 0$.
	\item $A\geq 0,\ M/A=C-B^*A^{\dagger}B\geq 0,\ range(B)\subseteq range(A)$.
	\item $C\geq 0,\ M/C=A-BC^{\dagger}B^*\geq 0,\ range(B^*)\subseteq range(C)$.
\end{enumerate}
     Here $A^{\dagger}$ and $C^{\dagger}$ refer to the Moore-Penrose pseudo inverses of A and C, respectively.
\end{lem}
\begin{rem}
\label{Rem2.9}
  Recall some properties of the Moore-Penrose pseudo inverse $A^{\dagger}$ of a matrix $A$ (see \cite[pages 29-30]{inverse} ):
  \begin{enumerate}[label=P\arabic*.]
    \item $AA^{\dagger}A=A$,\ $A^{\dagger}AA^{\dagger}=A^{\dagger}$.
    \item $(AA^{\dagger})^*=AA^{\dagger},\ (A^{\dagger}A)^*=A^{\dagger}A$.
    \item $AA^{\dagger}$ is the orthogonal projector onto the range of $A$,\ $A^{\dagger}A$ is the orthogonal projector onto the range of $A^*$.
    \item If $A$ is invertible, then $A^{\dagger}=A^{-1}$.
    \item If $A\geq0$, then $A^{\dagger}\geq0$.
    \end{enumerate}
\end{rem}
\begin{proof}(of Theorem 2.7)
  Since the $k$-positive map $\phi\neq0$, with respect to the canonical matrix units $E_{ij},\ i,j=1,...,m$, in $M_{m}(\mathbb{C})$, there exists an index $k\in \{1,2,...,m\}$ such that $\phi(E_{kk})\neq0$. Otherwise if $\phi(E_{kk})=0$ for every $k=1,..,m$, then $\phi(I_{m})=0$. Meanwhile for every $A\in M_{m}(\mathbb{C})^+$,\  $||A||I_{m}-A\geq0$ yields that $0=||A||\phi(I_{m})\geq\phi(A)$, implying $\phi=0$, which contradicts $\phi\neq0$. Without loss of generality, we assume that $\phi(E_{mm})\neq0$. Decompose the Choi matrix $C_{\phi}$ for $\phi$, with $A_{ij}=\phi(E_{ij}),\ i,j=1,...,m$, as follows:

\scalebox{0.8}{%
  \begin{minipage}{0.0\linewidth}
\begin{align*}
C_{\phi}&=
\left(
\begin{array}{ccccc}
A_{11} & \cdots &\cellcolor{blue!20}A_{1j} & \cdots & A_{1m} \\
\vdots & \ddots &\vdots & \ddots & \vdots \\
\cellcolor{blue!20}A_{i1} & \cdots \cellcolor{blue!20}&A_{ij} & \cdots & A_{im} \\
\vdots & \ddots &\vdots & \ddots & \vdots \\
A_{m1} & \cdots & \cellcolor{blue!20}A_{mj} & \cdots & A_{mm} \\
\end{array}
\right)\\
&=\left(
\begin{array}{ccccc}
A_{1m}A_{mm}^{\dagger}A_{m1} & \cdots & \cellcolor{blue!20}A_{1m}A_{mm}^{\dagger}A_{mj} & \cdots & A_{1m}A_{mm}^{\dagger}A_{mm} \\
\vdots & \ddots & \vdots & \ddots & \vdots \\
\cellcolor{blue!20}A_{im}A_{mm}^{\dagger}A_{m1} &\cdots & \cellcolor{blue!20} A_{im}A_{mm}^{\dagger}A_{mj} & \cdots & \cellcolor{blue!20}A_{im}A_{mm}^{\dagger}A_{mm} \\
\vdots & \ddots & \vdots & \ddots & \vdots \\
A_{mm}A_{mm}^{\dagger}A_{m1} &\ddots & \cellcolor{blue!20}A_{mm}A_{mm}^{\dagger}A_{mj} & \cdots & A_{mm}A_{mm}^{\dagger}A_{mm} \\
\end{array}
\right)\\
&+
\left(
\begin{array}{ccccc}
A_{11}-A_{1m}A_{mm}^{\dagger}A_{m1} & \cdots & \cellcolor{blue!20}A_{1j}-A_{1m}A_{mm}^{\dagger}A_{mj} & \cdots & A_{1m}-A_{1m}A_{mm}^{\dagger}A_{mm} \\
\vdots & \ddots & \vdots & \ddots & \vdots \\
\cellcolor{blue!20}A_{i1}-A_{im}A_{mm}^{\dagger}A_{m1} &\cdots & \cellcolor{blue!20} A_{ij}-A_{im}A_{mm}^{\dagger}A_{mj} & \cdots & \cellcolor{blue!20}A_{im}-A_{im}A_{mm}^{\dagger}A_{mm} \\
\vdots & \ddots & \vdots & \ddots & \vdots \\
A_{m1}-A_{mm}A_{mm}^{\dagger}A_{m1} &\ddots & \cellcolor{blue!20}A_{mj}-A_{mm}A_{mm}^{\dagger}A_{mj} & \cdots & A_{mm}-A_{mm}A_{mm}^{\dagger}A_{mm} \\
\end{array}
\right)
\\[3pt]
&\triangleq U+R=C_{\psi}+C_{\gamma}.
\end{align*}
   \end{minipage}}

For $i,j=1,...,m$, the $(i,j)$-entry of the matrix $U$ is given by $A_{im}A_{mm}^{\dagger}A_{mj}$, and the $(i,j)$-entry of the matrix $R$ is given by $R_{ij}=A_{ij}-A_{im}A_{mm}^{\dagger}A_{mj}$. Note that
\begin{align*}
U=\left(\begin{array}{c}
A_{1m}\\
\vdots\\
A_{im}\\
\vdots\\
A_{mm}
\end{array}
\right)A_{mm}^{\dagger}
\left(\begin{array}{ccccc}
A_{m1} & \cdots & A_{mj} & \cdots &A_{mm}
\end{array}
\right)\geq 0
\end{align*}
and $U\neq0$, since its $(m,m)$-entry is $A_{mm}A_{mm}^{\dagger}A_{mm}=A_{mm}=\phi(E_{mm})\neq0$. Then the map $\psi\neq0$ corresponding to the matrix $U$ is completely positive. By employing $k$-positivity of $\phi$, for arbitrary column vectors $w^1,w^2,..,w^{k-1}\in \mathbb{C}^m$, taking $\xi=(w^1,...,w^{k-1},e_m)^{T}$ where $e_{m}=(0,...,0,1)^{T}\in\mathbb{C}^m$,

\scalebox{0.9}{%
  \begin{minipage}{0.0\linewidth}
\begin{align*}
&\xi\xi^*=\left(
\begin{array}{ccccc}
w^1(w^1)^* &\cdots & w^1(w^j)^* & \cdots & w^1e_m^* \\
\vdots & \ddots & \vdots & \ddots & \vdots\\
w^i(w^1)^* &\cdots & w^i(w^j)^* & \cdots & w^ie_m^* \\
\vdots & \ddots & \vdots & \ddots & \vdots \\
e_m(w_1)^*&\cdots  & e_m(w^j)^* & \cdots & e_me_m^* \\
\end{array}\right)\geq0\\[2pt]
  &\Longrightarrow
  (id_k\otimes\phi(\xi\xi^*))=\left(
  \begin{array}{ccccc}
\phi(w^1(w^1)^*) &\cdots & \phi(w^1(w^j)^*) & \cdots & \phi(w^1e_m^*) \\
\vdots & \ddots & \vdots & \ddots & \vdots\\
\phi(w^i(w^1)^*) &\cdots & \phi(w^i(w^j)^*) & \cdots & \phi(w^ie_m^*) \\
\vdots & \ddots & \vdots & \ddots & \vdots \\
\phi(e_m(w_1)^*)&\cdots  & \phi(e_m(w^j)^*) & \cdots & \phi(e_me_m^*) \\
\end{array}\right)\geq0.
\end{align*}
   \end{minipage}
}
\\
By Lemma 2.8 (3), the condition $(id_k\otimes \phi)(\xi\xi^*)\geq0$ expands to:
\\

\scalebox{0.8}{%
  \begin{minipage}{0.0\linewidth}
\begin{align*}
&\ \ \ \ \ \ \  \left(
  \begin{array}{ccc}
\phi(w^1(w^1)^*) & \cdots & \phi(w^1(w^{k-1})^*) \\
\vdots & \ddots & \vdots \\
\phi(w^{k-1}(w^1)^*) & \cdots & \phi(w^{k-1}(w^{k-1})^*) \\
\end{array}\right)
\geq
\left(\begin{array}{c}
\phi(w^1e_m^*)\\
\vdots\\
\phi(w^{k-1}e_m^*)
\end{array}
\right)\phi(e_me_m^*)^{\dagger}\left(\begin{array}{ccc}\phi(e_m(w^1)^*)& \cdots & \phi(e_m(w^{k-1})^*)\end{array}\right)\\[2pt]
&\Longleftrightarrow
\left(
  \begin{array}{ccc}
\phi(w^1(w^1)^*)-\phi(w^1e_m^*)\phi(e_me_m^*)^{\dagger}\phi(e_m(w^1)^*) & \cdots & \phi(w^1(w^{k-1})^*)-\phi(w^1e_m^*)\phi(e_me_m^*)^{\dagger}\phi(e_m(w^{k-1})^*) \\
\vdots & \ddots & \vdots \\
\phi(w^{k-1}(w^1)^*)-\phi(w^{k-1}e_m^*)\phi(e_me_m^*)^{\dagger}\phi(e_m(w^1)^*) & \cdots &\phi(w^{k-1}(w^{k-1})^*)-\phi(w^{k-1}e_m^*)\phi(e_me_m^*)^{\dagger}\phi(e_m(w^{k-1})^*) \\
\end{array}\right)
\geq0.
\end{align*}
   \end{minipage}
}
\\[3pt]

For the $(s,t)$ entry in the above matrix, by linearity,
\begin{align*}
&\phi(w^se_m^*)\phi(e_me_m^*)^{\dagger}\phi(e_m(w^t)^*)\\
&=\phi\bigg(\sum_{i=1}^{m}w^s_ie_ie_m^*\bigg)\phi(e_me_m^*)^{\dagger}\phi\bigg(\sum_{j=1}^{m}\overline{w^t_j}e_me_j^*\bigg)\\
&=\bigg(\sum_{i=1}^{m}w^s_i\phi(E_{im})\bigg)\phi(E_{mm})^{\dagger}\bigg(\sum_{j=1}^{m}\overline{w^t_j}\phi(E_{mj})\bigg)\\
&=\sum_{i=1}^{m}\sum_{j=1}^{m}w^s_i\overline{w^t_j}\bigg(\phi(E_{im})\phi(E_{mm})^{\dagger}\phi(E_{mj})\bigg)\\
&=
\sum_{i=1}^{m}\sum_{j=1}^{m}w^s_i\overline{w^t_j}(A_{im}A_{mm}^{\dagger}A_{mj})\\
&=
\sum_{i=1}^{m}\sum_{j=1}^{m}w^s_i\overline{w^t_j}U_{ij}\\
&=
\sum_{i=1}^{m}\sum_{j=1}^{m}w^s_i\overline{w^t_j}\psi(e_ie_j^*)\\
&=
\psi(w^s(w^t)^*).
\end{align*}

Since $\gamma=\phi-\psi$, one has
\begin{align*}
\left(
  \begin{array}{ccc}
\gamma(w^1(w^1)^*)& \cdots & \gamma(w^1(w^{k-1})^*)\\
\vdots & \ddots & \vdots \\
\gamma(w^{k-1}(w^1)^*) & \cdots & \gamma(w^{k-1}(w^{k-1})^*)\\
\end{array}\right)
\geq0, \ \forall w^1,...,w^{k-1}\in \mathbb{C}^m,
\end{align*}
proving that $\gamma$ is $(k-1)$-positive. Moreover, all the entries of the $m^{th}$ row and $m^{th}$ column of the matrix $R$ are zero matrices. To show this, recall that $\phi$ is 2-positive$(k\geq2)$, hence any sub-block $\begin{pmatrix}\phi(E_{mm}) & \phi(E_{mj}) \\ \phi(E_{jm})& \phi(E_{jj})\end{pmatrix}\geq0$, for all $j=1,...,m-1$. By Lemma 2.8, one obtains that
%$ColumnSpace(\phi(E_{mj}))\subseteq ColumnSpace(\phi(E_{mm}))$, for all $j=1,..,m$.
$range(\phi(E_{mj}))\subseteq range(\phi(E_{mm}))$, for all $j=1,..,m$.
By property P3 in Remark~\ref{Rem2.9}, $A_{mm}A_{mm}^{\dagger}$ is the orthogonal projector onto the range of $A_{mm}$, so $R_{mj}=A_{mj}-A_{mm}A_{mm}^{\dagger}A_{mj}=0$, for all $j=1,...,m$. Denote the matrix $R=C_{\gamma}$ by:
\begin{align*}
R&=\left(
\begin{array}{cc|c}
&&0\\%\hline
\multicolumn{2}{c}{\raisebox{4pt}[2pt][0pt]{\Huge $K$}}&\vdots\\
0&\cdots&0\\
\end{array}
\right)=
\left(
\begin{array}{cc|c}
&&0\\%\hline
\multicolumn{2}{c}{\raisebox{4pt}[2pt][0pt]{\Huge $C_{\kappa}$}}&\vdots\\
0&\cdots&0\\
\end{array}
\right).
\end{align*}
Here, the map $\kappa\in B(M_{m-1}(\mathbb{C}),M_n(\mathbb{C}))$ is defined by the Choi matrix $K\in M_{(m-1)n}(\mathbb{C})$ through $\kappa(E_{st})=K_{st}, s,t=1,..,m-1$. It is obvious that $\gamma\in B(M_m(\mathbb{C}),M_n(\mathbb{C}))$ is the $m$-trivial lifting of $\kappa\in B(M_{m-1}(\mathbb{C}),M_n(\mathbb{C}))$. By Lemma 2.5, the map $\kappa\in B(M_{m-1}(\mathbb{C}),M_n(\mathbb{C}))$ is $(k-1)$-positive.
\end{proof}

A similar result holds for $k$-copositive maps.
\begin{cor}
  Let $\phi$ be a non-zero $k$-copositive $(2 \leq k < \min\{m,n\})$ map in $B(M_m(\mathbb{C}),M_n(\mathbb{C}))$. Then there exists a decomposition $\phi=\psi+\gamma$, where $\psi$ is a non-zero completely copositive map and $\gamma$ is a p-trivial lifting of a $(k-1)$-copositive map in $B(M_{m-1}(\mathbb{C}),M_n(\mathbb{C}))$, for some $p\in\{1,...,m\}$.
\end{cor}
\begin{proof}
  If $\phi$ is $k$-copositive, using the same arguments in proof of Theorem 2.7 for the matrix $\sum_{i,j=1}^{m}e_{ji}\otimes \phi(e_{ij})$, one obtains a decomposition $\sum_{i,j=1}^{m}e_{ji}\otimes \phi(e_{ij})=\sum_{i,j=1}^{m}e_{ji}\otimes \psi(e_{ij})+\sum_{i,j=1}^{m}e_{ji}\otimes \gamma(e_{ij})$, where $\psi$ is a non-zero completely copositive map and $\gamma$ is a $(k-1)$-copositive map which is a trivial lifting of a map in $B(M_{m-1}(\mathbb{C}),M_n(\mathbb{C}))$.
\end{proof}

\begin{thm}
  Let $2 \leq k < \min\{m,n\}$.
  Any non-zero $k$-positive (respectively $k$-copositive) map in $B(M_{m}(\mathbb{C}),M_{n}(\mathbb{C}))$ is the sum of at most $(k-1)$ many non-zero completely positive (respectively completely copositive) maps and a positive map which is the trivial lifting of a positive map in $B(M_{m-k+1}(\mathbb{C}),M_{n}(\mathbb{C}))$.
\end{thm}

\begin{proof}
  For a $k$-positive linear map $\phi$, repeatedly using Theorem 2.7 (respectively Corollary 2.10) until the remainder is a positive map.
\end{proof}

The Choi decomposition may no longer be valid for a general positive map $\phi$ even when $\phi$ is in $B(M_2(\mathbb{C}),M_2(\mathbb{C}))$. And it may not necessarily give us an algorithm to decompose a positive map in $B(M_2(\mathbb{C}),M_2(\mathbb{C}))$ as the sum of a completely positive map and a completely copositive map. Let us illustrate this
 by a simple example in $B(M_2(\mathbb{C}),M_2(\mathbb{C}))$.
\begin{exmp}
  Let $\epsilon$ be a real number and $\omega$ in $B(M_2(\mathbb{C}),M_2(\mathbb{C}))$ be defined through its Choi matrix:
\begin{align*}
C_{\omega}=&\left(
\begin{array}{cccc}
1 &0& 0 & \varepsilon \\
0 & 0 & \varepsilon & 0 \\
0 &\varepsilon & 0 & 0 \\
\varepsilon & 0& 0 & 1 \\
\end{array}\right),
\end{align*}
Hence the map $\omega$ is given by
$$\omega\begin{bmatrix}
                a & b \\
                c & d
              \end{bmatrix}=\begin{bmatrix}
                a & \varepsilon(b+c) \\
                \varepsilon(b+c) & d
              \end{bmatrix},\quad a,b,c,d\in\mathbb{C}.$$
For $\omega$ to be positive, it suffices to show for any vector $y=(y_1,y_2)^{T}\in \mathbb{C}^2$, the matrix
\begin{align*}
&|y_1|^2
\left(\begin{array}{cc}
1 &0 \\
0 & 0 \\
\end{array}\right)+
y_1\overline{y_2}
\left(\begin{array}{cc}
0 & \varepsilon \\
\varepsilon & 0 \\
\end{array}\right)+
y_2\overline{y_1}
\left(\begin{array}{cc}
0 & \varepsilon \\
\varepsilon & 0 \\
\end{array}\right)+
|y_2|^2
\left(\begin{array}{cc}
0 & 0 \\
0 & 1 \\
\end{array}\right)\\
&=
\left(\begin{array}{cc}
|y_1|^2 & 2\varepsilon Re(y_1\overline{y_2}) \\
2\varepsilon Re(y_1\overline{y_2}) & |y_2|^2 \\
\end{array}\right)
\end{align*}
is positive. This is equivalent to the condition that $-\frac{1}{2}\leq\varepsilon\leq\frac{1}{2}$. For Choi decomposition, using $A_{11}^{\dagger}=\left(
\begin{array}{cc}
1 &0   \\
0 & 0 \\
\end{array}\right)$, we have
\begin{align*}
C_{\omega}=\left(
\begin{array}{cccc}
1 &0& 0 & \varepsilon \\
0 & 0 & 0 & 0 \\
0 &0 & 0 & 0 \\
\varepsilon & 0& 0 & \varepsilon^2 \\
\end{array}\right)+
\left(
\begin{array}{cccc}
0 &0& 0 & 0\\
0 & 0 & \varepsilon & 0 \\
0 &\varepsilon & 0 & 0 \\
0 & 0& 0 & 1-\varepsilon^2 \\
\end{array}\right),
\end{align*}
and using $A_{22}^{\dagger}=\left(
\begin{array}{cc}
0 &0  \\
0 & 1 \\
\end{array}\right)$, we have
\begin{align*}
C_{\omega}=
\left(
\begin{array}{cccc}
\varepsilon^2 &0& 0 & \varepsilon \\
0 & 0 & 0 & 0 \\
0 &0 & 0 &0\\
\varepsilon & 0& 0 & 1 \\
\end{array}\right)+
\left(
\begin{array}{cccc}
1-\varepsilon^2  &0& 0 & 0\\
0 & 0 & \varepsilon & 0 \\
0 &\varepsilon & 0 & 0 \\
0 & 0& 0 &0 \\
\end{array}\right).
\end{align*}
In each of the two equations above, the last matrix corresponds to a linear map which is not positive. Meanwhile to decompose the map $\omega$ as the sum of a completely positive map and a completely copositive map, one splits the original matrix as follows:
\begin{align*}
C_{\omega}=
\left(
\begin{array}{cccc}
1/2 &0& 0 & \varepsilon \\
0 & 0 & 0 & 0 \\
0 &0 & 0 & 0 \\
\varepsilon & 0& 0 & 1/2 \\
\end{array}\right)+
\left(
\begin{array}{cccc}
1/2 &0& 0 & 0\\
0 & 0 & \varepsilon & 0 \\
0 &\varepsilon & 0 & 0 \\
0 & 0& 0 & 1/2 \\
\end{array}\right).
\end{align*}
Obviously under this splitting, the second and the third matrix in the above equation correspond to a completely positive map $\psi_1$ and a completely copositive map $\psi_2$, respectively, where $\psi_1\begin{bmatrix}
                a & b \\
                c & d
              \end{bmatrix}=\begin{bmatrix}
               \frac{a}{2} & \varepsilon b \\
                \varepsilon c & \frac{d}{2}
              \end{bmatrix}$ and $\psi_2\begin{bmatrix}
                a & b \\
                c & d
              \end{bmatrix}=\begin{bmatrix}
                \frac{a}{2} & \varepsilon c \\
                \varepsilon b & \frac{d}{2}
              \end{bmatrix}$.
\end{exmp}

\section{A reduced situation in $B(M_3(\mathbb{C}),M_3(\mathbb{C}))$}
In low dimensional cases such as $B(M_2(\mathbb{C}),M_2(\mathbb{C}))$ and $B(M_2(\mathbb{C}),M_3(\mathbb{C}))$, Woronowicz and St{\o}rmer respectively showed that every positive map is decomposable (see \cite{wor}). In this section, we will show that in $B(M_3(\mathbb{C}),M_3(\mathbb{C}))$, although positive maps may not be decomposable, 2-positive maps are always decomposable. Let us start with a useful lemma. For any $p\in\{1,...,m\}$, we assume that $\tilde{\chi}_p\in B(M_{m}(\mathbb{C}),M_n(\mathbb{C}))$ is the $p$-trivial lifting of a positive map $\chi\in B(M_{m-1}(\mathbb{C}),M_n(\mathbb{C}))$.
\par

\begin{lem}
  If $\chi$ is decomposable in $B(M_{m-1}(\mathbb{C}),M_n(\mathbb{C}))$, then its trivial lifting $\tilde{\chi}_p$ is also decomposable in $B(M_m(\mathbb{C}),M_n(\mathbb{C}))$.
\end{lem}
\begin{proof}
  Given a decomposable map $\chi\in B(M_{m-1}(\mathbb{C}),M_n(\mathbb{C}))$, then $\chi=\chi^1+\chi^2$, where $\chi^1$ is completely positive and $\chi^2$ is completely copositive. By Lemma 2.5, one obtains a completely positive map $\widetilde{\chi^1}_p$ and a completely copositive map $\widetilde{\chi^2}_p$ through $p$-trivial lifting of $\chi^1$ and $\chi^2$, respectively. By linearity of the trivial lifting, $\tilde{\chi}_p=(\widetilde{\chi^1+\chi^2})_p=\widetilde{\chi^1}_p+\widetilde{\chi^2}_p$ is decomposable in $B(M_m(\mathbb{C}),M_n(\mathbb{C}))$.
\end{proof}

The next result gives an affirmative answer to Conjecture 1.1.
\begin{thm}
   Every 2-positive or 2-copositive map $\phi$ in $B(M_3(\mathbb{C}),M_3(\mathbb{C}))$ is decomposable.
\end{thm}
\begin{proof}
  Without loss of generality, we assume the 2-positive(respectively 2-copositive) map $\phi$ is not zero. In this concrete case of $B(M_3(\mathbb{C}),M_3(\mathbb{C}))$, the peel-off process yields that:
\begin{align*}
  	\phi=\psi+\tilde{\kappa}_p \ \textit{for some p}\in\{1,...,m\}
\end{align*}
  where $\psi$ is completely positive (respectively completely copositive) and $\tilde{\kappa}_p$ is a p-trivial lifting of a positive map $\kappa\in B(M_2(\mathbb{C}),M_3(\mathbb{C}))$. Since every positive map in $B(M_2(\mathbb{C}),M_3(\mathbb{C}))$ is decomposable in $B(M_2(\mathbb{C}),M_3(\mathbb{C}))$ (see \cite{wor}), by Lemma 3.1, the lifted map $\tilde{\kappa}_p$ is decomposable in $B(M_3(\mathbb{C}),M_3(\mathbb{C}))$. Hence, $\phi=\psi+\tilde{\kappa}_p$ is also decomposable.
\end{proof}
\begin{defn}
  A positive linear map in $B(M_{m}(\mathbb{C}),M_n(\mathbb{C}))$ is called atomic if it is not the sum of a 2-positive map and a 2-copositive map.
\end{defn}
\begin{rem}
  From the definition, an atomic map in $B(M_{m}(\mathbb{C}),M_n(\mathbb{C}))$ is indecomposable. The converse is true when $m=n=3$.	
\end{rem}

\begin{cor}
  Every indecomposable map in $B(M_3(\mathbb{C}),M_3(\mathbb{C}))$ is atomic.
\end{cor}

\begin{cor}
  Under the dual cone correspondence (see \cite{facial}), one can completely determine the set inclusion relations in $B(M_3(\mathbb{C}),M_3(\mathbb{C}))$ as follows:
$$\begin{array}{ccccccl}
\mathbb{V}_1 & \subsetneqq &\mathbb{T} & \subsetneqq &  \mathbb{V}_2 & \subsetneqq &\mathbb{V}_{3}=(M_3(\mathbb{C})\otimes M_3(\mathbb{C}))^+\\
\updownarrow &            & \updownarrow &            &    \updownarrow &       & \updownarrow                \\
\mathbb{P}_1 & \supsetneqq & \mathbb{D} & \supsetneqq &  \mathbb{P}_2 &\supsetneqq & \mathbb{P}_{3}\cong(M_3(\mathbb{C})\otimes M_3(\mathbb{C}))^+\\
\end{array}$$
Here we denote by $\mathbb{V}_k$ the set of all quantum states of Schmidt number k, $\mathbb{P}_k$ the set of all k-positive maps,
$\mathbb{D}$ the cone of all decomposable maps and $\mathbb{T}$ the cone of all positive partial transpose states.
\end{cor}

\begin{rem}
	The inclusion $\mathbb{T} \subset \mathbb{V}_2$ gives an affirmative answer to Conjecture 1.2, stating that all positive partial transpose entangled states in $3\times3$ system have Schmidt number 2.
\end{rem}

\begin{exmp}
We will illustrate Choi decomposition using the 2-positive generalized Choi maps $\Phi[a,b,c]$  defined in \cite{gchoi} by
\begin{align*}
\Phi[a,b,c](X)=
\left(\begin{array}{ccc}
ax_{11}+bx_{22}+cx_{33} & -x_{12} & -x_{13}\\
-x_{21} & cx_{11}+ax_{22}+bx_{33} & -x_{23} \\
-x_{31} & -x_{32} & bx_{11}+cx_{22}+ax_{33}
\end{array}\right)
\end{align*}
for $X=[x_{ij}]\in M_{3}(\mathbb{C}^3)$, where $a,b,c,d$ are nonnegative numbers. Note that $\Phi[a,b,c]$ is 2-positive if and only if $a\geq2$ or $[1\leq a<2]\wedge[bc\geq(2-a)(b+c)]$. Let us consider the non-trivial case when the map $\Phi[a,b,c]$ is 2-positive but not completely positive. Hence $a\in [1,2)$ and $bc\geq(2-a)(b+c)$, which imply that
$$(*)\quad \ \ \ bc\geq4(2-a)^2\geq(a-1-\frac{2}{a})^2.$$
The Choi matrix of the map $\Phi[a,b,c]$ is
\\[1pt]

\scalebox{0.9}{%
  \begin{minipage}{0.0\linewidth}
\begin{align*}
C_{\Phi[a,b,c]}=
\left(\begin{array}{ccccccccc}
a &0 & 0& 0& -1& 0& 0&0 &-1\\
 0&c & 0& 0&0 & 0& 0&0 &0\\
0 & 0&b &0 & 0 &0 &0 & 0&0\\
0&0 &0 &b &0 &0 & 0&0 &0\\
-1 &0 &0 &0 & a& 0& 0&0 &-1\\
  0 &0 & 0&0 & 0& c& 0&0 &0\\
0& 0& 0& 0& 0& 0&c &0 &0\\
0 &0 &0 & 0&0 &0 &0 &b &0\\
-1 &0 & 0&0 & -1&0 &0 &0 &a\\
\end{array}\right) .
\end{align*}   \end{minipage}
}
\\[3pt]

Using $A_{11}^{\dagger}\triangleq \Phi[a,b,c](E_{11})^{\dagger}$, the Choi decomposition for the $2$-positive map $\Phi[a,b,c]$ is
\\[1pt]

\scalebox{0.9}{%
  \begin{minipage}{0.0\linewidth}
\begin{align*}
C_{\Phi[a,b,c]}=
\left(\begin{array}{ccc|ccc|ccc}
a &0 & 0& 0& -1&0 &0 & 0&-1\\
 0&c &0 &0 &0 & 0&0 &0 &0\\
 0&0 &b &0 & 0 &0 & 0&0 &0\\
0& 0&0 &0 & 0& 0& 0&0 &0\\
-1 &0 &0 & 0& \frac{1}{a}&0 & 0& 0&\frac{1}{a}\\
 0  &0 & 0& 0& 0& 0& 0&0 &0\\
0&0 &0 &0 & 0& 0&0 &0 &0\\
0 &0 &0 &0 &0 &0 &0 & 0&0\\
-1 &0 & 0&0 & \frac{1}{a}&0 &0 &0 &\frac{1}{a}\\
\end{array}\right)
+
\left(\begin{array}{ccc|ccc|ccc}
 0&0 &0 &0 & 0&0 &0 &0 &0\\
0 &0 & 0& 0& 0& 0&0 &0 &0\\
0 &0 & 0& 0&0  &0 &0 &0 &0\\
0&0 & 0&b &0 & 0& 0& 0&0\\
0 &0 & 0& 0& a-\frac{1}{a}&0 &0 &0 &-1-\frac{1}{a}\\
 0  & 0&0 &0 &0 & c&0 & 0&0\\
0& 0&0 &0 &0 &0 &c & 0&0\\
 0&0 &0 &0 &0 & 0&0 &b &0\\
0 &0 &0 &0 & -1-\frac{1}{a}&0 &0 &0 &a-\frac{1}{a}\\
\end{array}\right).
\end{align*}\end{minipage}
}
\\[3pt]

By 1-trivial lifting, the last matrix can be regarded as the Choi matrix for a positive linear map $\Upsilon$ in $B(M_{2}(\mathbb{C}),M_3(\mathbb{C}))$. While every positive map in $B(M_{2}(\mathbb{C}),M_3(\mathbb{C}))$ is decomposable, one obtains a decomposition of $\Upsilon$ as follows:
\begin{align*}
C_{\Upsilon}&=
\left(\begin{array}{ccc|ccc}
b &0 & 0& 0& 0&0\\
0& a-\frac{1}{a}&0 &0 &0 &-1-\frac{1}{a}\\
0 &0 & c&0 & 0&0\\
0 &0 &0 &c & 0&0\\
0 &0 & 0&0 &b &0\\
0 & -1-\frac{1}{a}&0 &0 &0 &a-\frac{1}{a}\\
\end{array}\right)\\
&=
\left(\begin{array}{ccc|ccc}
b &0 & 0& 0& 0&0\\
0& a-\frac{1}{a}&0 &0 &0 &\frac{1}{a}-a\\
0 &0 & 0&0 & 0&0\\
0 &0 &0 &c & 0&0\\
0 &0 & 0&0 &0 &0\\
0 & \frac{1}{a}-a&0 &0 &0 &a-\frac{1}{a}\\
\end{array}\right)
+
\left(\begin{array}{ccc|ccc}
0 &0 & 0& 0& 0&0\\
0& 0&0 &0 &0 &a-1-\frac{2}{a}\\
0 &0 & c&0 & 0&0\\
0 &0 &0 &0 & 0&0\\
0 &0 & 0&0 &b &0\\
0 & a-1-\frac{2}{a}&0 &0 &0 &0\\
\end{array}\right).
\end{align*}
Combining with the Choi decomposition of $C_{\Phi[a,b,c]}$, we have
\\[1pt]

\scalebox{0.9}{%
  \begin{minipage}{0.0\linewidth}
\begin{align*}
C_{\Phi[a,b,c]}=
\left(\begin{array}{ccc|ccc|ccc}
a &0 & 0& 0& -1&0 &0 & 0&-1\\
 0&c &0 &0 &0 & 0&0 &0 &0\\
 0&0 &b &0 & 0 &0 & 0&0 &0\\
0& 0&0 &b & 0& 0& 0&0 &0\\
-1 &0 &0 & 0& a&0 & 0& 0&\frac{2}{a}-a\\
 0  &0 & 0& 0& 0& 0& 0&0 &0\\
0&0 &0 &0 & 0& 0&c &0 &0\\
0 &0 &0 &0 &0 &0 &0 & 0&0\\
-1 &0 & 0&0 & \frac{2}{a}-a&0 &0 &0 &a\\
\end{array}\right)
+
\left(\begin{array}{ccc|ccc|ccc}
 0&0 &0 &0 & 0&0 &0 &0 &0\\
0 &0 & 0& 0& 0& 0&0 &0 &0\\
0 &0 & 0& 0&0  &0 &0 &0 &0\\
0&0 & 0&b &0 & 0& 0& 0&0\\
0 &0 & 0& 0& 0&0 &0 &0 &a-1-\frac{2}{a}\\
 0  & 0&0 &0 &0 & c&0 & 0&0\\
0& 0&0 &0 &0 &0 &0 & 0&0\\
 0&0 &0 &0 &0 & 0&0 &b &0\\
0 &0 &0 &0 & a-1-\frac{2}{a}&0 &0 &0 &0\\
\end{array}\right).
\end{align*}\end{minipage}
}
\\[3pt]

Hence $\Phi[a,b,c]=\Phi_1+\Phi_2$, where
\begin{align*}
\Phi_1&\begin{bmatrix}
                x_{11} & x_{12}& x_{13} \\
                x_{21} & x_{22}&x_{23} \\
                x_{31} & x_{32} & x_{33}
              \end{bmatrix}=\begin{bmatrix}
                 ax_{11}+bx_{22}+cx_{33}& -x_{12}&-x_{13}\\
                -x_{21}& cx_{11}+ax_{22}&(\frac{2}{a}-a)x_{23}\\
                -x_{31}& (\frac{2}{a}-a)x_{32}& bx_{11}+ax_{33}
              \end{bmatrix},\\
  \Phi_2&\begin{bmatrix}
                x_{11} & x_{12}& x_{13} \\
                x_{21} & x_{22}&x_{23} \\
                x_{31} & x_{32} & x_{33}
              \end{bmatrix}=\begin{bmatrix}
                 0& 0&0\\
                0& bx_{33}&(a-1-\frac{2}{a})x_{23}\\
               0 & (a-1-\frac{2}{a})x_{32}& cx_{22}
              \end{bmatrix}.
\end{align*}
Since $\Phi[a,b,c]$ is $2$-positive but not completely positive, by condition $(*)$ the matrices $C_{\Phi_1}$ and partial transpose of $C_{\Phi_2}$ are positive, implying that $\Phi_1$ is completely positive and $\Phi_2$ is completely copositive, respectively. Note that our method of writing the $2$-positive map $\Phi[a,b,c]$ as a sum of a completely positive map and a completely copositive map differs from another method mentioned 
in \cite[proof of Theorem 3.4]{gchoi}:
$$\Phi[a,b,c]=(1-\sqrt{bc})\Phi\bigl[\frac{a-\sqrt{bc}}{1-\sqrt{bc}},0,0\bigl]+\sqrt{bc}\Phi\bigl[1,\sqrt{\frac{b}{c}},\sqrt{\frac{c}{b}}\bigl].$$
So the decomposition of the $2$-positive map $\Phi[a,b,c]$ into a sum of a completely positive map and a completely copositive map is not unique.
\end{exmp}

In view of the previous results, it is natural to pose the following
\begin{quest}
\label{Quest3.9}
Does there exist a 2-positive but indecomposable map in $B(M_3(\mathbb{C}),M_4(\mathbb{C}))$?
\end{quest}
Most likely this question has an affirmative answer.

\section{Acknowledgement}
The authors are indebted to Professor M.-D. Choi for his sharing of the decomposition of the Choi matrix for a 2-positive map \cite{choi2}.
Subsequently, Professor M.-D. Choi and the authors independently proved Conjecture 1.1.
This work is partially supported by Singapore Ministry of Education Academic Research Fund Tier 1 Grant (No. R-146-003-193-112).

\end{document}